\documentclass[manuscript, screen]{acmart}
\usepackage{booktabs} 

\usepackage[ruled]{algorithm2e} 

\acmJournal{CIE}
\acmVolume{}
\acmNumber{}
\acmArticle{}
\acmYear{2017}
\acmMonth{8}
\usepackage{amsmath}
\newcommand*{\MyDef}{\mathrm{def}}

\newcommand*{\eqde}{\ensuremath{\mathop{\overset{\MyDef}{=}}}}
\newcommand\calF{\mathcal{F}}
\newcommand\calW{\mathcal{W}}
\newcommand\calS{\mathcal{S}}
\newcommand\esp[1]{{\mathchoice{\besp{#1}}{\sesp{#1}}{\sesp{#1}}{\sesp{#1}}}}
\newcommand\besp[1]{\mathbb{E}\left[#1\right]}
\newcommand\sesp[1]{\mathbb{E}[#1]}
\newcommand\Proba[1]{{\mathchoice{\bProba{#1}}{\sProba{#1}}{\sProba{#1}}{\sProba{#1}}}} 
\newcommand\bProba[1]{\mathbb{P}\left[#1\right]}
\newcommand\sProba[1]{\mathbb{P}[#1]}

\usepackage{enumitem}
\renewlist{cases}{enumerate}{10}
\setlist[cases]{label=\textbf{Case \arabic*}, itemindent=0pt,leftmargin=45pt,  listparindent=1.5 em}

\begin{document}

\title{A new analysis of Work Stealing with latency}

\author{Nicolas Gast}
\orcid{}
\email{nicolas.gast@imag.fr}
\affiliation{%
  \institution{Univ. Grenoble Alpes, Inria, CNRS, LIG}
  \streetaddress{}
  \city{Grenoble}
  \state{}
  \postcode{38400}
  \country{France}
}
\author{Mohammed Khatiri}
\orcid{1234-5678-9012-3456}
\email{mohammed.khatiri@imag.fr}
\affiliation{%
  \institution{Univ. Grenoble Alpes, CNRS, Inria, LIG, France -  }
  \institution{University Mohammed First, Faculty of Sciences, LaRI, Morocco}
  \streetaddress{}
  \city{Grenoble - Oujda}
  \postcode{38400 - 60000}
}
\author{Denis Trystram}
\orcid{}
  \email{denis.trystram@imag.fr}
\affiliation{%
  \institution{Univ. Grenoble Alpes, CNRS, Inria, LIG}
  \streetaddress{}
  \city{Grenoble}
  \state{}
  \postcode{38400}
  \country{France}
  }

\author{Fr\'ed\'eric Wagner}
\orcid{}
\email{frederic.wagner@imag.fr}  
\affiliation{%
  \institution{Univ. Grenoble Alpes, CNRS, Inria, LIG}
  \streetaddress{}
  \city{Grenoble}
  \state{}
  \postcode{38400}
  \country{France}
}
\renewcommand\shortauthors{Gast et al.}

\begin{abstract}
  We study in this paper the impact of communication latency on the
  classical \emph{Work Stealing} load balancing algorithm.  Our paper
  extends the reference model in which we introduce a latency
  parameter. By using a theoretical analysis and simulation, we study
  the overall impact of this latency on the Makespan (maximum completion time).  
  We derive a new expression of the expected running time of a \textit{bag of tasks} scheduled
    by Work Stealing. This expression enables us to predict under which
  conditions a given run will yield acceptable performance.  For
  instance, we can easily calibrate the maximal number of processors
  to use for a given work/platform combination.  
  All our results are validated through simulation on a wide range of parameters.
\end{abstract}

%
%

\ccsdesc[300]{Theory of computation~Scheduling algorithms}

%

\keywords{Work Stealing, Latency, Makespan Analysis}

%

\maketitle

\section{Introduction}
\label{sec:intro}

The motivation of this work is to study how to extend the analysis of the Work Stealing (WS) algorithm
in a distributed-memory context, where communications matter.
WS is a classical on-line scheduling algorithm proposed for
shared-memory multi-cores~\cite{Arora2001} whose  principle is recalled in the next section.
As it is common, we target the minimization of the \textit{Makespan},
defined as the maximum completion time of the parallel application.
We present a theoretical analysis for an upper bound of the expected makespan
and we run a complementary series of simulations in order to assess how this new bound behaves in practice
depending on the value of the latency.

\subsection{Motivation for studying WS with latency}
Distributed-memory clusters consist in independent processing elements with private local memories linked by an interconnection network.
In such architectures, communication issues are crucial, they highly influence the performances of the applications~\cite{MCMCA}.
However, there are only few works dealing with optimized allocation strategies and the relationships with the allocation and scheduling process is most often ignored.
In practice, the impact of scheduling may be huge since the whole execution can be highly affected by a large communication latency of interconnection networks~\cite{MultiCoreClusterArchitectur2015}.
Scheduling is the process which aims at determining where and when to execute the tasks of a target parallel application.
The applications are represented as directed acyclic graphs where the vertices are the basic operations and the arcs are the dependencies between the tasks~\cite{CosnardTrystram92}.
Scheduling is a crucial problem which has been extensively studied under many variants for the successive generations of parallel and distributed systems.
The most commonly studied objective is to minimize the makespan (denoted by $C_{\max}$) and the underlying context is usually to consider centralized algorithms.
This assumption is not always realistic, especially if we consider distributed memory allocations and an on-line setting.

WS is an efficient scheduling mechanism targeting medium range parallelism of multi-cores for fine-grain tasks.
Its principle is briefly recalled as follows:
each processor manages its own (local) list of tasks.
When a processor becomes idle it randomly chooses another processor and steals some work (if possible).
Its analysis is probabilistic since the algorithm itself is randomized.
Today, the research on WS is driven by the question on how to extend the analysis for the characteristics of new computing platforms (distributed memory, large scale, heterogeneity).
Notice that beside its theoretical interest, WS has been implemented successfully in several languages and parallel libraries including Cilk~\cite{Leiserson1998,Leiserson2009},
TBB (Threading Building Blocks)~\cite{Robison2008}, the PGAS language~\cite{Dinan2009,Seung-Jai2011} and
the KAAPI run-time system~\cite{Kaapi2007}.

\subsection{Related works}
\label{subsec:relatedworks}

We start by reviewing the most relevant theoretically-oriented works.
WS has been studied originally by Blumofe and Leiserson in~\cite{Blumofe1999}.
They showed that the expected Makespan of a series-parallel precedence graph
with $\mathcal{W}$ unit tasks on $p$ processors is bounded
by $E(C_{\max}) \leq \frac{\mathcal{W}}{p}+\mathcal{O}(D)$ where $D$ is the length of the
critical path of the graph (its depth).
This analysis has been improved in Arora \textit{et al.}~\cite{Arora2001} using potential functions. 
The case of varying processor speeds has been studied by Bender and Rabin in~\cite{Bender2002}
where the authors introduced a new policy called \textit{high utilization scheduler} that extends the homogeneous case. 
The specific case of tree-shaped computations with a more accurate model has been studied in~\cite{Sanders1999}. 
However, in all these previous analyses, the precedence graph is constrained to have only one source and an out-degree of at most 2 which does not easily model
the basic case of independent tasks. 
Simulating independent tasks with a binary precedences tree gives a
bound of $\frac{\mathcal{W}}{p}+\mathcal{O}(\log_2 (\mathcal{W}))$
since a complete binary tree of $\mathcal{W}$ vertices has a depth
$D \leq \log_2 (\mathcal{W})$.  However, with this approach, the
structure of the binary tree dictates which tasks are stolen.  In
complement,~\cite{Gast2010} provided a theoretical analysis based on a
Markovian model using mean field theory. They targeted the expectation
of the average response time and showed that the system converges to a
deterministic Ordinary Differential Equation.  Note that there exist
other results that study the steady state performance of WS when the
work generation is random including Berenbrink \textit{et
  al.}~\cite{Berenbrink2003}, Mitzenmacher~\cite{Mitzenmache1998},
Lueling and Monien~\cite{Lueling1993} and Rudolph \textit{et
  al.}~\cite{Rudolph1991}.  More recently, in \cite{Denis2013},
Tchiboukjian \textit{et al.} provided the best bound known at this
time: $\frac{\mathcal{W}}{p}+c.(\log_2 \mathcal{W})+\Theta(1)$ where
$c\approx3.24$.

In all these previous theoretical results, communications are not directly addressed 
(or at least are taken implicitly into account by the underlying model).
WS largely focused on shared memory systems and 
its performance on modern platforms (distributed-memory systems, hierarchical plateform, clusters with explicit communication cost) is not really well understood. 
The difficulty lies in the problems of communication which become more crucial in modern platforms~\cite{Arafat2016}.

Dinan \textit{et al.}~\cite{Dinan2009} implemented WS on large-scale clusters,
and proposed to split each tasks queues into a part accessed asynchronously by local processes 
and a shared portion synchronized by a lock which can be access by any remote processes in order to reduce contention.
Multi-core processor based on NUMA (non-uniform memory access) architecture is the mainstream today and such new platforms include accelerators as GPGPU~\cite{Yang2017}. 
The authors propose an efficient task management mechanism which is to divide the application into a large
number of fine grained tasks which generate large amount of small communications between the CPU and the GPGPUs.
However, these data transmissions are very slow.
They consider that the transmission time of small data and big data is the same.

Besides the large literature on theoretical works, there exist more practical studies implementing WS libraries where some attempts were provided for taking into account communications:
SLAW is a task-based library introduced in~\cite{Gue2010}, combining work-first
and help-first scheduling policies focused on locality awareness 
in PGAS (Partitioned Global Address Space) languages like UPC (Unified Parallel C). 
It has been extended in HotSLAW, which provides a high level API that abstracts concurrent task management~\cite{Seung-Jai2011}.
\cite{Shigang2013} proposes an asynchronous WS (AsynchWS) strategy which exploits opportunities to overlap communication with
local tasks allowing to hide high communication overheads in distributed memory systems. 
The principle is based on a hierarchical victim selection, also based on PGAS.
Perarnau and Sato presented in~\cite{swann2014} an experimental evaluation of WS on the scale of ten thousands compute nodes where the communication depends on the distance between the nodes.
They investigated in detail the impact of the communication on the performance. In particular, the physical distance between remote nodes is taken into account.
Mullet \textit{et al.} studied in~\cite{Muller2106} 
Latency-Hiding, a new WS algorithm that hides the overhead
caused by some operations, such as waiting for a request from a client or waiting
for a response from a remote machine.
The authors refer to this delay as \emph{latency} which is slightly
different that the more general concept we consider in our paper. 
Agrawal \textit{et al.} proposed an analysis~\cite{Agrawal} showing the optimality for task graphs with bounded degrees
and developed a library in Cilk++ called \textit{Nabbit} for executing tasks
with arbitrary dependencies, with reasonable block sizes.





\subsection{Contributions}

In this work, we study how communication latency impacts work
stealing. Our work has three main contributions. First, we create a
new realistic scheduling model for distributed-memory clusters
of $p$ identical processors including latency denoted by $\lambda$.
Second, we provide an upper bound of the expected
makespan.  This bound is composed of the usual lower bound on the best
possible load-balancing~$\frac{\mathcal{W}}{p}$ plus an additional
term proportional to $\lambda\log_2(\frac{\mathcal{W}}{\lambda})$
where $\mathcal{W}$ and $p$ are the total amount of work and the total number of processors respectively.
Third, we provide simulation results to assess this bound.  These
experiments show that the theoretical bound is roughly 5 times greater
than the one observed in the experiments but that the additional term
has indeed the form $c\lambda\log_2(\frac{\mathcal{W}}{\lambda})$. The
theoretical analysis shows that $c<16.12$ while the simulation results
suggest that $c\approx3.8$. 

The analysis is based on an adequate potential function.
There are two reasons that distinguish this analysis in regard to the existing ones:
finding the right function (the natural extension does not work since we now need to consider
in transit work).
Its property is that it should diminish after any steal related operation.
We also consider large timesteps of duration equal to the communication latency.



\section{Work-Stealing}
\label{sec:wsmechanisms}

In this section, we introduce some formal notations and the WS algorithm.
Finally we present the variant of the WS algorithm that we analyse.

\subsection{Notation and definition}

We consider a discrete time model. There are $p$ processors. We denote
by $\calW_{i}(t)\in\{0,1,\dots\}$
 the amount of work on processor $P_{i}$ at time $t$ (for
 $i\in\{1\dots p\}$). At unit of work
corresponds to one unit of execution time.  We denote the total amount
of work on all processors by $\calW(t) = \sum_{i=1}^{p}
\calW_{i}(t)$.
At $t=0$ all work is on $P_{1}$. The total amount of work at time $0$
is denoted by $\calW = \calW_{1}(0)$.

\subsection{WS algorithm}

Work Stealing is a decentralized list scheduling algorithm
where each processor $P_i$ maintains its own local queue $Q_i$ of tasks to execute. 
$P_i$ uses $Q_i$ to get and execute tasks while $Q_i$ is not empty. 
When $Q_i$ becomes empty $P_i$ chooses another processor $P_j$ randomly
and sends a steal request to it. When $P_j$ receives this request,
he answers by either sending some of its work or by a fail response.
We will see bellow the case of the fail response.

The answer of $P_j$ can be to transfer
some of its work or a negative response. 

We analyse one of the variants of the WS algorithm that has the
following features: 

\begin{itemize}

    \item \textbf{Latency:} All communication takes a time $\lambda$
        that we call the latency. A work request that is sent at time $t-\lambda$
        by a thief will be received at time $t$ by the victim. The thief
        will then receive an answer at time $t+\lambda$. As we consider a discrete-time model,
        we say that a work request arrives at time $t$ if it arrives between
        $t-1$ (not-included) and $t$.  This means that at time $t$, this work
        request is treated.  The number of incoming work requests at time $t$
        is denoted by $\mathcal{R}(t)$.  Note that
        $\forall t, 0 \leq \mathcal{R}(t) \leq p-1$.  It is equal to the
        number of processors sending a work request at time $t - \lambda$.

        When a processor $P_i$ receives a work request from a thief $P_j$, it
        sends a part of its work to $P_j$. This communication takes again
        $\lambda$ units of time.  $P_j$ receives the work at time $t+\lambda$.
        We denote by $\mathcal{S}_{i}(t)$ the amount of work in transit from
        $P_i$ at time $t$.  At end of the communication $\mathcal{S}_{i}$
        becomes 0 until a new work request arrives.
    \item \textbf{Steal Threshold:} The main goal of WS is to share work
        between processors in order to balance load and speed-up execution.
        In some cases however it might be beneficial to keep work local and
        answer negatively to some steal requests.  We assume that if the
        victim has less than $\lambda$ unit of work to execute, the steal
        request fails.
    \item \textbf{Single work transfer:} We assume that a processor can
        send some work to at most one processor at a time.  While the
        processor sends work to a thief it replies by a fail response to any
        other steal request.  Using this variant, the steal request may fail
        in the following cases: when the victim does not have enough
        work or when it is already sending some work to another thief.
        Another case might happen when the victim receives more than one
        steal request at the same time. He deals a random thief and
        send a negative response to the remaining thieves.
    \item \textbf{Work division:} in order to keep a balanced WS, we
        consider that the victim sends to the thief a part of work such that
        both loads will be balanced at the end of the communication.  More
        precisely, if $P_i$ receives a work request from $P_j$ at time $t$
        then:
        \begin{equation}
            \label{equ1}
            \calW_i(t) = \frac{\calW_i(t-1) -1 + \lambda}{2} \mbox{ and } \mathcal{S}_i(t) = \frac{\calW_i(t-1) -1 - \lambda}{2}
        \end{equation}
        After a time $\lambda$:
        \begin{equation}
            \label{equ2}
            \calW_{i}(t+\lambda) = \calW_i(t) - \lambda =
            \frac{\calW_{i}(t-1) -1 -\lambda}{2} = \mathcal{S}_i(t) =
            \calW_{j}(t+\lambda)
        \end{equation}
\end{itemize}

\section{Theoretical Analysis}
\subsection{Principle}

Before presenting the detailed analysis, we first describe its main steps.

We denote by $C_{\max}$ the makespan (\emph{i.e.}, total execution
time).  In a WS algorithm each processor either executes work or tries
to steal work. As the round-trip-time of a communication is $2\lambda$
and the total amount of work is equal to $\calW$ and the number of processors is $p$, we have
$pC_{\max} \le \calW + 2\lambda\#\textit{StealRequests}$ where $p$ is the number of processors.
We therefore have a straightforward bound of the Makespan:
\begin{equation}
  \label{Cmax}
  \mathcal{C}_{max}  \leq \frac{\calW}{p} + 2\lambda\frac{\#StealRequests}{p}
\end{equation}
Note that the above inequality is not an equality because the
executing might end while some processors are still waiting for work.

Our analysis makes use of a pontential function that represents how
well the jobs are balanced in the system.  We bound the
number of steal requests by showing that each event involving a steal
operation contributes to the decrease of this potential function. 
This analysis shares some similarities with the one of \cite{Denis2013}
as we make use of a potential function that decreases with steal
requests.  The key difficulty to apply these ideas in our case is that
communications take $\lambda$ time units. At first, it seems that
longer durations should translate linearly into the time taken by
steal requests but this would neglect the fact that longer steal
durations reduce the number of steal requests.  

In order to analyse the impact of $\lambda$, we reconsider the time
division as periods of duration $\lambda$.  We analyse the system at
each time step $k\lambda$ for $k\in\{0,1,2,\dots\}$.
To simplify the notations between the time $t$ and the interval division
$k\lambda$, we denote respectively by $w_i(k)$ and $s_i(k)$ the quantities
$\calW_i(k\lambda)$ and $\mathcal{S}_i(k\lambda)$.  We denote by
$\phi_i(k)$ the part of the potential linked to processor~$i$ and by
$\phi(k) = \sum_i \phi_i(k)$ the potential.  We also define the total
number of incoming steal work requests in the interval
$(\lambda(k-1), \lambda k] $ by
$r_k = \sum_{j=1}^{\lambda}\mathcal{R}((k-1)\lambda+j)$ and we denote
by $q(r_k)$ the probability that a processor receives one or more
requests in the interval $(\lambda(k-1), \lambda k]$ (this function
will be computed in the next section).

In the next section, we analyze the decrease of $\phi(k)$ as a
function of the number of steal requests.
We show that there exists a function
$h: \{1 \cdots p\} \rightarrow [0,1]$ depending on the number of
incoming steal requests $r_k$ in the time interval
$(\lambda(k-1), \lambda k]$, such that in average, $\phi(k+1)$ is less
than $h(r_k)\phi(k)$.  Finally, we use this to derive a bound on the
total number of steal requests. By using equation \eqref{Cmax}, we
obtain a bound on the Makespan.

\section{Detailed Analysis}

In this section, we prove the main result of the paper, which is a
bound on the total completion time $C_{\max}$ and is summarized by the
following theorem :
\begin{theorem}
    \label{theo:cmax}
  Let $C_{max}$ be the Makespan of $\mathcal{W} = n$ unit independent tasks scheduled by WS with latency algorithm.
    Then,

  \begin{align*}
      (i)  \qquad &\esp{C_{max}} \leq \frac{\mathcal{W}}{p} +  4\lambda\gamma\log_2(\frac{\mathcal{W}}{2\lambda}) + 3\lambda \\
      (ii) \qquad  &\Proba{ C_{max}  \ge \frac{\mathcal{W}}{p} +  4\lambda\gamma\log_2(\frac{\mathcal{W}}{2\lambda}) + 3\lambda + x} \le {2}^{-x}
    \label{expdiff}                                                                
  \end{align*}
    In particular:
  \begin{align*}
      (ii) \qquad  &\esp{C_{max}} \leq \frac{\mathcal{W}}{p} +  16.12\lambda\log_2(\frac{\mathcal{W}}{2\lambda}) + 3\lambda
  \end{align*}
\end{theorem}

The proof of this result is based on the analysis of the decrease of a
function -- that we call the \emph{potential}. The potential at
time-step $k\lambda$ is denoted $\phi(k)$ and is defined as 
\begin{equation*}
  \phi(k) = \sum_{i=1}^{p} \phi_i(k), \qquad \text{ where $\phi_i(k) =
    w_{i}(k)^{2} + 2s_{i}(k)^{2}$ \quad for $i\in\{1\dots,p\}$}.
\end{equation*}
This potential function always decreases.  It is maximal when all work
is contained in one processor which is the potential function at time
$0$ and is equal to $\phi(0) = \mathcal{W}^2$.  The schedule completes
when the potential becomes $0$.

We divide our proof of Theorem~\ref{theo:cmax} in two lemmas. First,
in Lemma~\ref{lem:decrease} we show that the expected decrease of the
potential can be bounded as a function of the number of work
requests. Second, in Lemma~\ref{lem:OfR}, we show how such a bound
leads to a bound on the expected number of work requests. 

We denote by $\mathcal{F}_{k}$ all events up to the interval
$((k-1)\lambda, k\lambda]$.
\begin{lemma}
  \label{lem:decrease}
  For all steps $k$, the expected ratio between $\phi(k+1)$ and $\phi(k)$ knowing $\mathcal{F}_{k}$ is bounded by: 
  \begin{equation}
    \mathbb{E}[\phi(k+1) \:|\: \mathcal{F}_{k}] \leq \left(1-\frac{q(r_k)}{4}\right)\phi(k) 
    \label{expdiff}
\end{equation}

    Where $q(r_k)$ is the probability for a processor to receive one
    or more requests in the interval $(\lambda(k-1),\lambda k]$ knowing that there are $r_k$ incoming steal requests.
\end{lemma}

\begin{proof} 
  To analyze the decrease of the potential function, we distinguish
  different cases depending on whether the processor is executing
  work, sending or answering steal requests.  We show that each case
  contributes to a variation of potential.

  Between time $k\lambda$ and $(k+1)\lambda$, a processor $P_i$ does
  one the following things:
  \begin{enumerate}
  \item The processor is executing and sending work.
  \item The processor is executing work and available to respond to the steal requests sent by idle processors.
  \item The processor is executing work and will be idle soon.
  \item The processor is idle.
  \end{enumerate}
  We analyse below how much each case contributes to the decrease of
  the potential.
  \begin{cases}
  \item \textbf{:} $P_i$ is executing and started to send work to
    another idle processor $P_j$ before time $k\lambda$.  This means
    that $P_j$ will receive work before time $(k+1)\lambda$.  As we do
    not know when the communication finishes, we study the worst case
    in this scenario. In particular, we make as if (a) $P_i$ and $P_j$
    respond negatively to any steal request before the end of
    communication; and (b) they do not execute work.  Such events
    would decrease the potential function.
    
    By the ''work-division'' principle (Equation~\eqref{equ1}), once
    the steal request is completed we have
    $\calW_j(t)=\calW_i(t)=\calS_i(t)=s_i(k)$. This shows that the
    quantities of work and of work in transit at time $k+1$ satisfies
    \begin{equation*}
      w_i(k+1) \leq s_i(k) \mbox{ and } w_j(k+1) \leq s_i(k)
    \end{equation*}
    \begin{equation*}
      s_i(k+1) = 0 \mbox{ and } s_j(k+1) = 0 
    \end{equation*}
    Thus,
    \begin{equation*}                                                              
        \phi_i(k+1) = w_i(k+1)^2 + 2 s_i(k+1)^2 \leq w_i(k)^2                                      
    \end{equation*}
    \begin{equation*}                                                              
        \phi_j(k+1) = w_i(k+1)^2 + 2 s_i(k+1)^2 \leq w_i(k)^2                                      
    \end{equation*}
    Moreover by Equation~\eqref{equ2} we have $s_i(k)\le w_i(k)$.
    This shows that $\phi_i(k)=w_i(k)^2+2s_i(k)^2\ge
    3s_i(k)^2$.
    Hence, the ratio of potential each couple ($P_i$ as victim, $P_j$
    as thief) is less than~:
    \begin{equation*}                                                       \frac{\phi_i(k+1) + \phi_j(k+1)}{\phi_i(k) + \phi_j(k)} \leq  \frac{ 2s_i(k)^2 }{ 3s_i(k)^2 } \leq \frac{2}{3}
    \end{equation*}
  \item \textbf{:} $P_i$ is executing work and available to respond to
    a steal request. We distinguish two cases: (case 2a) if $P_i$
    receives a requests or (case 2b) if it does not receive a request.
    
    \textbf{Case 2a} -- We compute the ratio of potential between
    between $k\lambda$ and $(k+1)\lambda$ when $P_i$ receives one or
    more steal requests.  $P_i$ will respond to the first steal
    request. All other steal requests will fail.  The worst case is to
    receive the steal request in the end of the interval. This shows
    that
    \begin{equation*}
      w_i(k+1) \leq \frac{(w_i(k)-\lambda)+\lambda}{2} - 1 \mbox{ and } s_i(k+1) \leq \frac{(w_i(k) - \lambda)-  \lambda}{2},
    \end{equation*}
    which implies that
    \begin{equation*}
      w_i(k+1)^2 \leq \left(\frac{w_i(k)}{2} \right) ^2 \leq \frac{w_i(k)^2}{4}  
      \qquad \text{and}\qquad s_i(k+1)^2 = \frac{(w_i(k) - 2\lambda)^2}{4} \leq \frac{w_i(k)^2}{4}
    \end{equation*}
    This generates a ratio of potential between $k\lambda$ and
    $(k+1)\lambda$ for each couple ($P_i$ as victim, $P_j$ as thief)
    smaller than~:
    \begin{equation*}
      \frac{\phi_i(k+1) + \phi_j(k+1)}{\phi_i(k) + \phi_j(k)} \leq \frac{w_i(k+1)^2+2s_i(k+1)^2}{w_i(k)^2} \leq \frac{3w_i(k)^2}{4w_i(k)^2} \leq \frac{3}{4}  
    \end{equation*}
    
    \textbf{Case 2b} -- If $P_i$ does not receive any work requests,
    the work decreases,
    in which case 
    $\phi_i(k+1)/\phi_i(k) \le 1$.
    
 Let $q(r_k)$ be the probability that the
    processor $P_i$ receives a work request between $k\lambda$ and
    $(k+1)\lambda$.  To compute $q(r_k)$, we observe that $P_i$
    receives zero work requests if the $r_k$ thieves choose another
    processor.  Each of these events is independent and happens with
    probability $\frac{p-2}{p-1}$.  Hence, the probability that $P_i$
    receives one or more work requests is~:
    \label{Procavailable}
    \begin{equation}
      q(r_k) = 1 - \left(\frac{p-2}{p-1} \right)^{r_k}
      \label{proba}
    \end{equation}
    This shows that the ratio of potential in this scenario is:
    \begin{equation*}
      \frac{ \phi_i(k+1)} { \phi_i(k)} \leq \frac{3}{4}q(r_k) + (1-q(r_k)) \leq 1 - \frac{q(r_k)}{4}
    \end{equation*} 
    
\item  \textbf{:} $P_i$ with little amount of work $w_i(k) \leq \lambda$ and $s_i(k) = 0$,
in this case $P_i$ will respond negatively to any work requests and 
the potential function goes to $0$ and generates a ratio equal to $0$
\label{littleWork}

\item 
$P_i$ is idle, so it can be thief in $case 1$ which waits for work or thief in $case 2$ which sends a steal request.
    In both cases we already have taken its contribution to the potential into account.
\end{cases}

Using the variation of each these scenarios we find that the expected potential time $k+1$ is bounded by: 
\begin{equation*}
    \mathbb{E}[\phi(k+1)\:|\: \mathcal{F}_{k}] \leq \frac{2}{3}\sum_{i \in case1}\phi_i(k) + \left(1-\frac{q(r_k)}{4}\right)\sum_{i \in case2}\phi_i(k) + 0\times\sum_{i \in case3} \phi_i(k)
\end{equation*}

\begin{equation*}
    \mathbb{E}[\phi(k+1)\:|\: \mathcal{F}_{k}] \leq \max\left(\frac{2}{3},1-\frac{q(r_k)}{4},0\right)\left( \sum_{i \in case1}\phi_{i}(k)+ \sum_{i \in case2}\phi_i(k) + \sum_{i \in case3}\phi_i(k)\right ) 
  \end{equation*}
  Thus,
  \begin{equation*}
    \mathbb{E}[\phi(k+1)\:|\: \mathcal{F}_{k}] \leq
    \max\left(\frac{2}{3},1-\frac{q(r_k)}{4}\right)\phi(k) =
    \left(1-\frac{q(r_k)}3\right)\phi(k),
  \end{equation*}
  where the last inequality holds because $q(r_k)\le1$ and therefore $1-\frac{q(r_k)}{4} \geq
  \frac{3}{4}$.
\end{proof}

\begin{lemma}
  \label{lem:OfR}
  Assume that there exists a function $h$:
  ${0 \cdots p} \rightarrow [0, 1]$ such that the expected potential
  at time $k+1$ given $\mathcal{F}_{k}$ satisfies:
  \begin{equation}
    \mathbb{E}[\phi(k+1) \:|\: \mathcal{F}_{k}] \leq h(r_k)\phi(k) 
    \label{eq:expdiff}
  \end{equation}
  Let $\phi(0)$ denote the potential at time $0$ and let $\gamma$ be
  defined as:
  \begin{equation*}
    \gamma \eqde \max_{1 \leq r \leq p } \frac{r}{ -p \log_2(h(r))  } 
  \end{equation*} 
  
  Let $\tau$ be the first time step at which $w_i(\tau)$ is less than
  $3\lambda$ for all processors:
  \begin{align*}
  \tau \eqde \min \{k \mathrm{~s.t.~} \forall i \in \{1\dots p\} :
  w_i(k) \leq 3\lambda \}.    
  \end{align*}

 The number of incoming steal requests until
  $\tau$ is, $R = \sum_{s=0}^{\tau-1} r_s$, satisfies:
  \begin{align*}
      (i)\qquad  &\esp{R} \leq 2p\gamma\log_2(\frac{\mathcal{W}}{\lambda})  \\
      (ii)\qquad &\Proba{R\ge 2p\gamma\log_2(\frac{\mathcal{W}}{\lambda}) + x} \le {2}^{-x}
    \label{expdiff}                                                                
  \end{align*}
\end{lemma}

\begin{proof} 
  By definition of $\gamma$, for a quantity $r_k$, we have
  $\log_2(h(r_k)) \leq \frac{r_k}{ -p \gamma }$ and therefore
  $h(r_k)\le 2^{-r_k/(p\gamma)}$. By using
  Equation~\eqref{eq:expdiff}, this shows that
  \begin{align*}
    \esp{\phi(k+1) \mid \calF_k}&\le h(r_k)\phi(k)\\
    &\le \phi(k)2^{-r_k/(p\gamma)}. 
  \end{align*}
  Let $X_{k}=\phi(k)\prod_{i=0}^{k-1}2^{r_i/(p\gamma)}$. By using the
  equation above, we have:
  \begin{align*}
    \esp{X_{k+1}\mid\calF_k} &= \esp{\phi(k+1) \mid \calF_k}
                               \prod_{i=0}^{k}2^{r_i/(p\gamma)}\\
                             &\le \phi(k) 2^{-r_k/(p\gamma)}
                               \prod_{i=0}^{k}2^{r_i/(p\gamma)}=X_k
  \end{align*}
  This shows that $(X_k)_k$ is a submartingale for the filtration
  $\calF$. As $\tau$ is a stopping time for the filtration $\calF$,
  Doob's optional stopping theorem (see \emph{e.g.},
  \cite[Theorem~4.1]{durrett1996probability}) implies that
  \begin{equation}
    \esp{X_\tau} \le \esp{X_0}. 
    \label{eq:doob}
  \end{equation}
  By definition of $X$, we have $X_0=\phi(0)$ and
  $X_\tau=\phi(\tau)2^{R/(p\gamma)}$. Hence, this implies that
  \begin{align}
    \esp{\phi(\tau)2^{R/(p\gamma)}} \le \phi(0). 
    \label{eq:doob2}
  \end{align}
  
  Recall that $\tau$ is the first interval in which each processor has
  an amount of work less than~$3\lambda$, $\tau = \min$\{$k$:
  $\forall i \in [1,p] $ : $ w_i(k) \leq 3\lambda $\}.  This means
  that at $\tau-1$, there exists at least one processor $i$ with
  $w_i(\tau-1) > 3\lambda$.  If this processor received a steal
  request between $\tau-1$ and $\tau$, we have
  $w_i(\tau) \geq \lambda$. If this processor did not receive a steal
  request between $\tau-1$ and $\tau$, we have
  $w_i(\tau)>2\lambda$. This implies that
   \begin{equation*}  
     \phi(\tau) \ge\phi_i(\tau) \ge w_i^2(\tau) + 2s_i^2(\tau) \geq
     4\lambda^2 \qquad\mathrm{a.s.}
   \end{equation*}
   Plugging this into Equation~\eqref{eq:doob2} shows that 
   \begin{align}
     \esp{2^{R/(p\gamma)}} \le \frac{\phi(0)}{4\lambda^2}.
     \label{eq:doob3}
   \end{align}
   By Jensen's inequality (see \emph{e.g.},
   \cite[Equation~(3.2)]{durrett1996probability}), we have
   $\esp{R/(p\gamma)}\le\log_2\esp{2^{R/(p\gamma)}}$. This shows that
   \begin{align*}
     \esp{R} &\le p\gamma \log_2\frac{\phi(0)}{4\lambda^2}\\
             &= p \gamma \left(\log_2\phi(0) - 2\log_2(2\lambda)\right).
   \end{align*}
   As $\phi(0)\le\calW^2$, we have $\log_2(\phi(0))\le
   2\log_2\calW$. Hence:
   \begin{align}
     \esp{R} &\le 2p \gamma \left(\log_2\calW - \log_2(2\lambda)\right).
               \label{logratiophi}
   \end{align}
   By Markov's inequality, Equation~\eqref{eq:doob3} implies that for
   all $a>0$:
   \begin{align*}
     \Proba{2^{R/(p\gamma)}\ge a} \le \frac{\calW^2}{4a\lambda^2}.
   \end{align*}
   By using $a=\calW^2/(4\lambda^2)2^{x/(p\gamma)}$, this implies that
   \begin{align*}
       \Proba{R \ge 2p\gamma(\log_2 W - \log_2(2\lambda)) + x} 
     &=
       \Proba{2^{R/p\gamma} \ge (W/2\lambda)^2 2^{x/(p\gamma)}} \le  2^{-x/(p\gamma)}
   \end{align*}
 \end{proof}

 We are now ready to conclude the proof of Theorem~\ref{theo:cmax}, by
 applying the previous lemma with $h(r_k) = (1-\frac{q(r_k)}{4})$. 

 \begin{proof}[Proof of Theorem~\ref{theo:cmax}] The number of
   incoming steal requests until $\tau$ is bounded by
   Lemma~\ref{lem:OfR}~$(i)$.  Using the definition of $\tau$, there
   remain at most 3 steps of $\lambda$ to finish the execution.  By using
   Equation~\eqref{Cmax} we have,
\begin{align*}
    \esp{\mathcal{C}_{max}}  &\leq \frac{\calW}{p} + \frac{\esp{R}}{p}  2\lambda + 3\lambda 
                             &\leq \frac{\calW}{p} + 4\lambda\gamma\log_2(\frac{\mathcal{W}}{2\lambda}) + 3\lambda
\end{align*}
    By the same way, we use Lemma~\ref{lem:OfR}~$(ii)$ and equation \ref{Cmax} we find:  
    
\begin{align*}
      &\Proba{ C_{max}  \ge \frac{\mathcal{W}}{p} +  4\lambda\gamma\log_2(\frac{\mathcal{W}}{2\lambda}) + 3\lambda + x} \le {2}^{-x}
\end{align*}


We now show that the above constant $\gamma$ satisfies
$\gamma<4.03$. For that, we first show that the quantity $r/-\log(1-q(r)/4)$ is increasing in $r$. 
Then, we use the maximum of $r$ and we bound the value $(p-1)/-\log(1-q(p-1)/4)$ by $(p-1)/(2-\log_2(3+1/e))$.

Let $f(r) \eqde -\log_2(1-q(r)/4)$ and $g(r)= r / f(r)$. By definition of $q(r)$, $f(r)$ can be written as: 
\begin{align*}
    f(r) = -\log_2\left( \frac{3}{4}  + \frac{1}{4}\left(1-\frac{1}{p-1}\right)^r \right) =  2 - \log_2\left( 3 + \left(1-\frac{1}{p-1}\right)^r \right)
\end{align*}
Denoting $v \eqde 1 - 1/(p-1)$ and $x \eqde v^r$, the derivative of $g$ with respect to $r$~is:  

\begin{align*}
    g'(r) &= \frac{ v^rr\ln(r) - v^r(\ln(3+v^r) - 2\ln(2)) - 3\ln(3+v^r) + 6\ln(2) }{(3+v^r)f(r)^2\ln(2)}\\ 
          &=  \frac{ x\ln(x) - x(\ln(3+x) - 2\ln(2)) - 3\ln(3+x) + 6\ln(2) }{(3+x)f(r)^2\ln(2)}\\ 
          &=  \frac{ x\ln(x) - \ln(3+x)(3 + x) + 2\ln(2)(x+3) }{(3+x)f(r)^2\ln(2)}\\ 
\end{align*}
As $x<1$, the derivative of $x\ln(x) - \ln(3+x)(3 + x) + 2\ln(2)(x+3)$  with regard to $x$ is $2\ln(2)-\ln(3+x)+\ln(x) < 0$.
This shows that $x\ln(x) - \ln(3+x)(3 + x) + 2\ln(2)(x+3) < \ln(1)-4\ln(3+1)+\ln(1) < 0 $, thus $g$ is decreasing. Using the fact that for all $ p \geq 2$: 

\begin{align*}
    \left( 1 - \frac{1}{p-1} \right)^{p-1} &= \exp \left((p-1)\ln(1-\frac{1}{p-1}) \right) 
    &\leq \exp \left(-(p-1)\frac{1}{p-1}) \right) = \frac{1}{e} 
\end{align*}

Then,
\begin{align*}
    \gamma = \max_{1 \leq r \leq p} \frac{1}{p}g(r) = \frac{1}{p}g(p-1) &\leq \frac{1}{ 2 - \log_2\left(3 + \left(1 - \frac{1}{p-1} \right)^{p-1}  \right)  } \\
           &\leq \frac{1}{ 2 - \log_2\left(3 + \frac{1}{e}  \right) } <  4.03 
\end{align*}

\end{proof}



\section{Experiments Analysis}
\label{sec:experiemnts}

In the previous section, we proved a new upper bound of the Makespan
of WS with an explicit latency.  The objective of this section is to
study WS experimentally in order to confirm the theoretical results
and to refine the constant $\gamma$. We developed to this end an \textit{ad-hoc}
simulator that follows strictly our model. 

We start by describing our WS simulator and the considered test configurations.
Using the experimental results, we show that our theoretical bound is close to the experimental results. 
We conclude with a discussion on these observations. 

\subsection{Simulator and configurations}
We have developed a python discrete event simulator for running adequate
experiments.  This simulator follows the model described in
Section~\ref{sec:wsmechanisms} to schedule an amount $\mathcal{W}$ of
work on a distributed platform composed of $p$ identical
processors.  Between each two processors, the communication cost is
equal to the latency $\lambda$.  These three parameters are
configurable and the simulator is generic enough to be used in
different contexts of online scheduling and interfaces with standard
trace analysis tools.  To ensure reproducibility, the code is
available on
github\footnote{\url{https://github.com/wagnerf42/ws-simulator}}.

Let us describe our experimental parameters.  
We consider constant speed processors, which means that the work can be described in a time unit basis, 
and the same holds for the latency.  
Ultimately, only the ratio between
$\frac{W}{p}$ and $\lambda$ matters.  Similar results would be
observed by multiplying $\frac{W}{p}$ and $\lambda$ by the same
constant.  For our tests we take different parameters with
$\mathcal{W}$ between $10^5$ and $10^8$, $p$ between 32 and 256 and
$\lambda$ between 2 and 500. Each experiment has been reproduced 1000
times.

\subsection{Validation of the bound}

As seen before, the bound of the expected Makespan consists of two
terms: the first term is the ratio $\mathcal{W}/{p}$ which does not
depend of the configuration and the algorithm, and the second term
which represents the overhead related to steal requests.  Our analysis
bounds the second term to derive our bound on the Makespan.

Therefore, to analyze its validity, we define the overhead ratio as
the ratio between the second term of our theoretical bound
($4\gamma\lambda\log_2(\mathcal{W}/\lambda)+3\lambda$) and the
execution time simulated minus the ratio $\mathcal{W}/p$.  We study
this overhead ratio under different parameters $\mathcal{W}$, $p$ and
$\lambda$.
\begin{figure}
  \includegraphics{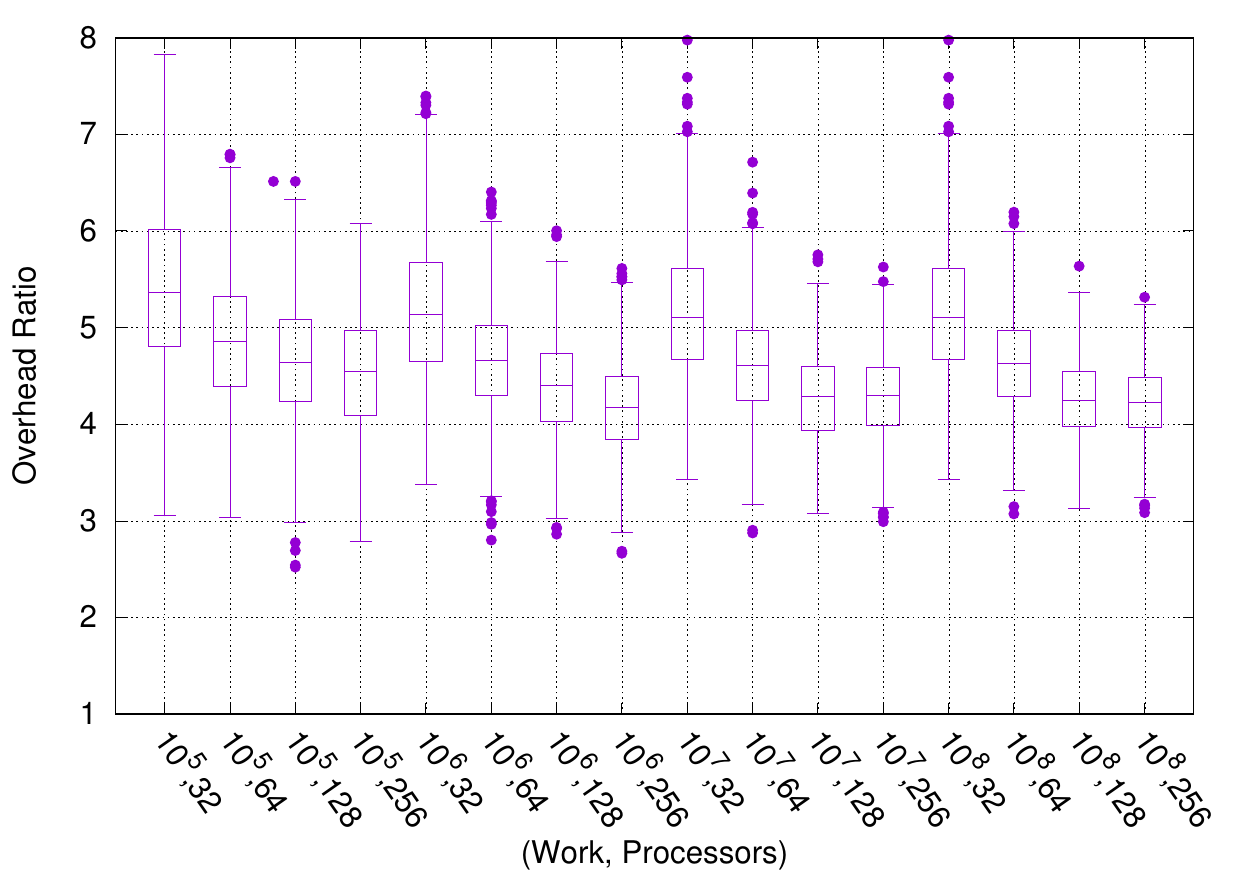}
  \caption{Overhead ratio as a function of $(\mathcal{W},p)$
    ($\lambda = 262$)}
  \label{fig:accuracy}
\end{figure}

Figure~\ref{fig:accuracy} plots the overhead ratio according to each
couple $(\mathcal{W}, p)$, for a latency of $262$ (1000 runs).
Similar observations have been observed with all values of latency
used.  The~x-axis is $(\mathcal{W}, p)$ for all values of
$\mathcal{W}$ and $p$ intervals and the y-axis shows the overhead
ratio.  We use here a BoxPlot graphical method to present the
results. BoxPlots give a good overview and a numerical summary of a
data set.  The “interquartile range” in the middle part of the plot
represents the middle quartiles where 50\% of the results are
presented.  The line inside the box presents the median. The whiskers
on either side of the IQR represent the lowest and highest quartiles
of the data.  The ends of the whiskers represent the maximum and
minimum of the data, and the individual points beyond the whiskers
represent outliers.

We observe that our bound is systematically about 4 to 5.5 times greater
to the one computed by simulation (depending on the range of parameters). The ratio between the two bounds
decreases with the number of processors but seems fairly independent
to $\calW$.

\subsection{Discussion}

The challenge of this paper is to analyze WS algorithm with an
explicit latency.  We presented a new analysis which derives a bound on the
expected Makespan for a given $\mathcal{W}$, $p$ and~$\lambda$.  
It shows that the expected Makespan is bounded by $\calW/p$ plus an
additional term bounded by $4\gamma\approx16$ times
$\lambda\log_2(\calW/(2\lambda))$.  As observed in
Figure~\ref{fig:accuracy}, the constant $4\gamma$ is about four
to five times larger than the one observed by simulation.  A more
precise fitting based on simulation results leads to the expression
$\mathcal{W}/p + 3.8\lambda\log_2(\mathcal{W}/\lambda)$.
We briefly review here the main steps of the proof involving
approximations and conclude with a simple fit of results on the best
constant matching the expected Makespan.

First, recall that our analysis relies on $r_k$, the number of steal
requests arriving at each time interval.  The exact values of all
$r_k$ are unknown and we rely on worst case majorations.  In the proof of
Theorem~\ref{theo:cmax} we define $\gamma$ as the maximal value of
$\frac{r}{-p\log_2(h(r))}$. The function
$r\mapsto\frac{r}{-p\log_2(h(r))}$ is about $2.8$ for small values of
$r$ and tends to $4.03$ when the number of work requests $r$ is
maximal. In the simulation, we observe that the number of work
requests is often low. Therefore, our definition of $\gamma$ probably
contributes a factor $4/2.8\approx1.4$ to the overhead ratio.

The second approximation done is while we computed the diminution
of the potential using the maximum diminution in all cases described
in the proof of Lemma~\ref{lem:decrease}.  This analysis could be
improved by taking a more complex potential function but this
will lead to much harder computations for marginal improvements.

The third approximation is that we assumed that we
do not know when exactly steal requests arrive in the interval.  We
therefore always took the worst case (arrivals at the end of the intervals). We believe that this
approximation has only a minor effect on the overhead ratio since it
mainly impacts the potential diminution obtained from computations (as
opposed to stealing).

Finally, the value of $\gamma$ depends on $p$. 
To achieve a constant bound, we consider again
the worst case obtained as $p$ tends to infinity.  This explains the
fact that the overhead ratio increases slightly when the number
of processors decreases ($4.5$ for $256$ processors and $5$ for $32$
processors).


\section{Conclusion}
\label{sec:conclusion}

We presented in this paper a new analysis of Work Stealing algorithm
where each communication has a communication latency of $\lambda$.  Our main result
was to show that the expected Makespan of a load of $W$ on a cluster of
$p$ processors is bounded by $W/p+16.12\lambda\log_2(W/(2\lambda))$. 

We based our analysis on potential functions to bound the expected
number of steal requests.  We therefore derived a theoretical bound on the
expected Makespan.  We also extend this analysis one step further, by
providing a bound on the probability to exceed the bound of the
Makespan.

To assess the tightness of this analysis we developed an
\textit{ad-hoc} simulator.  We showed by comparing the theoretical bound and
the experimental results that our bound is realistic.  We observed moreover that
our bound (established on worst case analysis) is 5 times greater than the experimental results 
and it is stable for all the tested values.

This work will certainly be the basis of incoming studies on more
complex hierarchical topologies where communications matter. As such, it is important as it allows
a full understanding of the behavior of various Work Stealing
implementations in a base setting.





\end{document}